\documentclass[runningheads]{llncs}

\usepackage {amssymb}
\usepackage {amsmath}
\usepackage [usenames] {color}
\usepackage {paralist}
\usepackage {mathptmx}

\usepackage {enumerate}
\usepackage {plaatjes}
\usepackage {cite}
\usepackage{float}
\usepackage [left,pagewise] {lineno}
\usepackage {xspace}
\usepackage {a4wide}

\let\doendproof\endproof
\renewcommand\endproof{~\hfill\qed\doendproof}

\definecolor {infocolor} {rgb} {0.6,0.6,0.6}

\definecolor {sepia} {rgb} {0.75,0.30,0.15}

\floatstyle{ruled}
\newfloat{algorithm}{thp}{alg}
\floatname{algorithm}{Algorithm}

\newcommand {\mathset} [1] {\ensuremath {\mathbb {#1}}}
\newcommand {\R} {\mathset {R}}

\newcommand {\script} [1] {\ensuremath {\mathcal {#1}}}

\newcommand {\eps} {\varepsilon}

\DeclareMathOperator {\ply}{\Delta}
\DeclareMathOperator {\nextstep}{next\_step}
\DeclareMathOperator {\picksensor}{pick\_sensor}

\newcommand{\marrow}{\marginpar[\hfill$\longrightarrow$]{$\longleftarrow$}}

\renewcommand{\remark}[3]{\textcolor{blue}{\textsc{#1 #2:}}
\textcolor{red}{\marrow\textsf{#3}}}
\renewcommand{\remark}[3]{\relax}
\newcommand{\maarten}[2][says]{\remark{Maarten}{#1}{#2}}

\newtheorem {observation}[theorem] {Observation}

\title{Tracking Moving Objects with Few Handovers}
\author{David Eppstein \and Michael T. Goodrich \and Maarten L\"offler}
\institute{Dept. of Computer Science, Univ. of California, Irvine}

\date{}

\begin{document}

\maketitle

\begin {abstract}
We study the online problem of assigning a moving point 
to a base-station region that contains it. 
For instance, the moving object could represent a cellular phone and the base 
station could represent the coverage zones of 
cell towers. 

Our goal is to minimize the number of \emph{handovers} that occur when 
the point moves outside its assigned
region and must be assigned to a new one. 
We study this problem in terms of 
a competitive analysis measured
as a function of $\ply$, the \emph{ply} of 
the system of regions, that is, the maximum number of regions
that cover any single point. 

In the offline version of this problem,
when object motions are known in advance, a simple greedy strategy 
suffices to determine an optimal assignment of objects to base stations, 
with as few handovers as possible. 
For the online version of this problem for
moving points in one dimension, we present a deterministic 
algorithm that achieves a competitive ratio of $O(\log\ply)$ with
respect to the optimal algorithm, 
and we show that no better ratio is possible. 
For two or more dimensions, we present a randomized online 
algorithm that achieves a competitive ratio of $O(\log\ply)$ with
respect to the optimal algorithm, and a deterministic algorithm that achieves a competitive ratio of $O(\ply)$;
again, we show that no better ratio is possible.
\end {abstract}



\section {Introduction}
A common problem in wireless sensor networks involves the online
tracking of moving objects~\cite{aekd-tfmt-10,cysa-atdp-05,%
gtzts-stpea-10,hh-ttqw-05,ppk-eqttt-03,yz-mdot-09,zjr-iddsc-02}. 
Whenever a moving object leaves a region corresponding to its
tracking sensor, a nearby sensor  must take over the job of tracking the object.
Similar \emph{handovers} are also used in cellular phone services to track moving customers~\cite{tekinay1991handover}.
In both the sensor tracking and cellular phone applications, handovers involve considerable 
overhead~\cite{gtzts-stpea-10,hh-ttqw-05,ppk-eqttt-03, tekinay1991handover,zjr-iddsc-02},
so we would like to minimize their number.

Geometrically, we can abstract the problem in terms of a set of $n$ closed regions in $\R^d$, for a constant $d$, which represent the sensors or cell towers.
We assume that any pair of regions 
intersects at most a constant number of times, as would be the case,
say, if they were unit disks (a common geometric
approximation used for wireless 
sensors~\cite{aekd-tfmt-10,cysa-atdp-05,gtzts-stpea-10,hh-ttqw-05,ppk-eqttt-03,zjr-iddsc-02}).
We also have one or more moving entities,
which are represented as points traveling along 1-dimensional curves (which we do not assume to  be smooth, algebraic, or otherwise well-behaved, and which may not be known or predictable by our algorithms) with a time stamp 
associated to each point on the curve (Figure~\ref{fig:example}).

  \eenplaatje[scale=0.8]{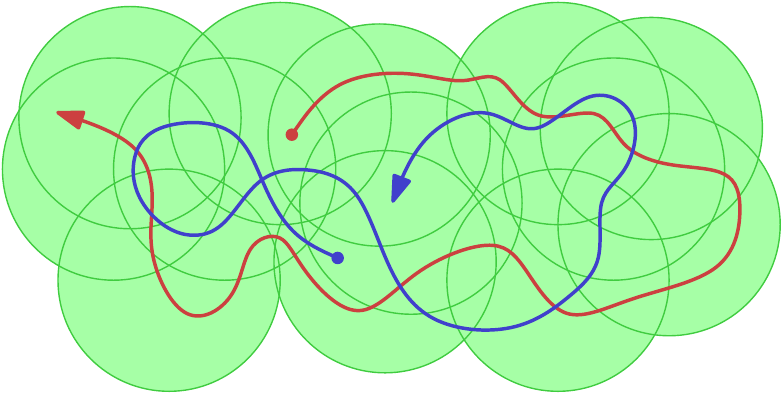} {\label{fig:example} Example input}

We need to track the entities via regions that respectively contain them; 
hence, for each moment in time, 
we must assign one of the regions to each entity, $p$,
with the requirement that $p$ is 
inside its assigned region at each moment in time. 
Finally,
we want to minimize the number of 
times that we must change
the assignment of the region tracking an entity, so as to minimize
the number of handovers.

We also consider a generalized version of this problem, where
each point $p$ is required to be assigned to $c$ regions at each 
moment in time.
This generalization is motivated by the need for 
\emph{trilateration} in cellular networks and wireless sensor 
networks~\cite{hellebrandt2002estimating,yl-qtcbi-10},
where directional information from three or more 
sensors is used to identify the coordinates of a moving point.
  
\subsection{Related Work}
There has been considerable previous work in the wireless sensor
literature on mobile object tracking. So, rather than providing a
complete review of this area, let us simply highlight some of the
most relevant work from the wireless sensor literature.

Cao {\it et al.}~\cite{cysa-atdp-05} study 
the problem of modeling an object moving along a straight-line
trajectory among uniformly-distributed 
unit-disk sensors. Their analysis involves a
probabilistic study of when tracking is feasible if sensors enter their
tracking states independently at random.
Alaybeyoglu {\it et al.}~\cite{aekd-tfmt-10}
also study the object tracking problem using uniformly-distributed 
unit disks to model sensors,
with their focus on the problem of identifying the tracking sensor
with strongest signal in each case.

Zhou {\it et al.}~\cite{zjr-iddsc-02} introduce the idea of using
handovers to reduce energy in mobile object tracking problems among
wireless sensor networks.
Pattem {\it et al.}~\cite{ppk-eqttt-03} study energy-quality trade-offs
for various strategies of mobile object tracking, including one
with explicit handovers.
He and Hou~\cite{hh-ttqw-05} likewise study mobile object
tracking with respect to handover minimization, deriving probabilistic
upper and lower bounds based on distribution assumptions about 
the moving objects and wireless sensors.
Ghica {\it et al.}~\cite{gtzts-stpea-10} study the 
problem of tracking an object among sensors modeled as unit disks 
so as to minimize handovers,
using probabilistic assumptions about the object's future location while
simplifying the tracking requirements to discrete epochs of time.

The analysis tool with which we characterize the performance of our
algorithms comes from research in 
\emph{online algorithms}, where problems are defined 
in terms of a sequence of decisions that must be made one at a time, before knowing the sequence of future requests.
Sleator and Tarjan~\cite{st-aelup-85}, introduce
\emph{competitive analysis},
which has been used for
a host of subsequent online algorithms (e.g., see~\cite{be-occa-98}).
In competitive analysis, 
one analyzes an online algorithm by comparing its
performance against that
of an idealized adversary, who can 
operate in an offline fashion, making his choices after seeing the
entire sequence of items.

We are not aware of any previous work that applies competitive analysis to the problem of 
handover minimization.
Nevertheless, this problem 
can be viewed from a computational geometry perspective as an instantiation of the 
\emph{observer-builder} framework of Cho {\it et al.}~\cite{cmp-mnntu-09},
which itself is related to the \emph{incremental motion} model of
Mount {\it et al.}~\cite{mnpsw-cfim-04}, the 
\emph{observer-tracker} model
of Yi and Zhang~\cite{yz-mdot-09}, and the well-studied 
\emph{kinetic data structures} 
framework~\cite{aeg-kbisd-98,bgsz-pekds-97,g-kdssar-98,ghsz-kcud-01}.
In terms of the observer-builder model,
our problem has an \emph{observer} who watches the motion 
of the point(s) we wish to track and a \emph{builder} who maintains
the assignment of tracking region(s) to the point(s).
This assignment would define a set of Boolean \emph{certificates}, which become \emph{violated} when a
point leaves its currently-assigned tracking region. The observer
would notify the builder of any violation, and the builder would use information about the current
and past states of the point(s) to make a new assignment (and define
an associated certificate).
The goal, as in the previous work by
Cho {\it et al.}~\cite{cmp-mnntu-09}, would be to minimize the
number of interactions between the observer and builder, as measured using
competitive analysis.
Whereas Cho {\it et al.}~apply their model to the maintenance of 
net trees for moving points, 
in our case the interactions to be minimized correspond to handovers, and our results supply the algorithms that would be needed to implement a builder for handover minimization.
Yi and Zhang~\cite{yz-mdot-09} study a general online tracking problem, but with a different objective function than ours: when applied to mobile object tracking, 
rather than optimizing the number of handovers, their scheme would
aim to minimize the distance between objects and the
base-station region to which they are each assigned.

Several previous
papers study overlap and connectivity problems for geometric
regions, often in terms of their \emph{ply}, the maximum number of regions that cover any point.
Guibas {\it et al.}~\cite{ghsz-kcud-01} study the
maintenance of connectivity information among moving unit
disks in the kinetic data structure framework.
Miller {\it et al.}~\cite{mttv-sspnn-97} introduce the concept of ply and show how sets of
disks with low ply possess small geometric separators.
Eppstein {\it et al.}~\cite{eg-snprn-08,egt-gortc-09} study road
network properties and algorithms using a model based on 
sets of disks with low ply after outliers are removed.
Van~Leeuwen~\cite{v-basdg-06} studies the
minimum vertex cover problem for disk
graphs, providing an asymptotic FPTAS for this problem on
disk graphs of bounded ply.
Alon and Smorodinsky~\cite{as-cfcsd-06} likewise study 
coloring problems for sets of disks with low ply.

Our problem can also be modeled as a \emph{metrical task system} in which the sensor regions are represented as states of the system, the cost of changing from state to state is uniform, and the cost of serving a request is zero for a region that contains the request point and two for other regions. Known randomized online algorithms for metrical task systems~\cite{IraSei-TCS-98} would give a competitive ratio of $O(\log n)$ for our problem, not as good as our $O(\log\ply)$ result, and known lower bounds for metrical task systems would not necessarily apply to our problem.

\subsection{New Results}
In this paper,
we study the problem of assigning moving points 
in the plane to containing base station regions in an online
setting and use the competitive analysis to characterize the
performance of our algorithms.  
Our optimization goal in these algorithms 
is to minimize the number of \emph{handovers} that occur when 
an object moves outside the range of its currently-assigned base 
station and must be assigned to a new base station. 
We measure the competitive ratio of 
our algorithms as a function of $\ply$, the \emph{ply} of 
the system of base station regions, 
that is, the maximum number of such regions
that cover any single point. 
When object motions are known in advance, 
as in the offline version of the object traking problem, 
a simple greedy strategy 
suffices to determine an optimal assignment of objects to base stations, 
with as few handovers as possible. 
For the online problem, on the other hand,
for moving points in one dimension, we present a deterministic online 
algorithm that achieves a competitive ratio of $O(\log\ply)$, 
with respect to the offline optimal algorithm,
and we show that no better ratio is possible. 
For two or more dimensions, we present a randomized algorithm 
that achieves a competitive ratio of $O(\log\ply)$, 
and a deterministic algorithm that achieves a competitive 
ratio of $O(\ply)$; again, we show that no better ratio is possible.

\section {Problem Statement and Notation} \label {sec:prelims}

Let $\script D$ be a set of $n$ regions in $\R^d$. These regions represent the areas that can be covered by a single sensor. We assume that each region is a closed, connected subset of $\R^d$ and that the boundaries of any two regions intersect $O(1)$ times -- for instance, this is true when each region is bounded by a piecewise algebraic curve in $\R^2$ with bounded degree and a bounded number of pieces. With these assumptions, the arrangement of the pieces has polynomial complexity $O(n^d)$. The \emph {ply} of $\script D$ is defined to be the maximum over $\R^d$ of the number of regions covering any point.  We always assume that $\script D$ is fixed and known in advance.

Let $T$ be the trajectory of a moving point in $\R^d$. We assume that $T$ is represented as a continuous and piecewise algebraic function from $[0,\infty)$ to $\R^d$, with a finite but possibly large number of pieces. We also assume that each piece of $T$ crosses each region boundary $O(1)$ times and that it is possible to compute these crossing points efficiently. We also assume that $T([0,\infty))\subset\cup\script D$; that is, that the moving point is always within range of at least one sensor; this assumption is not realistic, and we make it only for convenience of exposition. Allowing the point to leave and re-enter $\subset\cup\script D$ would not change our results since the handovers caused by these events would be the same for any online algorithm and therefore cannot affect the competitive ratio.

As output, we wish to report a \emph{tracking sequence} $S$: a sequence of pairs $(\tau_i,D_i)$ of a time $\tau_i$ on the trajectory (with $\tau_0=0$) and a region $D_i\in\script D$ that covers the portion of the trajectory from time $\tau_i$ to $\tau_{i+1}$. We require that for all $i$, $\tau_i<\tau_{i+1}$. In addition, for all $i$, it must be the case that $T([\tau_i,\tau_{i+1}])\subseteq D_i$, and there should be no $\tau'>\tau_{i+1}$ for which $T([\tau_i,\tau'])\subset D_i$; in other words, once a sensor begins tracking the moving point, it continues tracking that point until it moves out of range and another sensor must take over.  Our goal is to minimize $|S|$, the number of pairs in the tracking sequence. We call this number of pairs the \emph{cost} of $S$; we are interested in finding tracking sequences of small cost.

Our algorithm may not know the trajectory $T$ completely in advance. In the \emph{offline tracking problem}, $T$ is given as input, and we must find the tracking sequence $S$ that minimizes $|S|$; as we show, a simple greedy algorithm accomplishes this task. In the \emph{online tracking problem}, $T$ is given as a sequence of \emph{updates}, each of which specifies a single piece in a piecewise algebraic decomposition of the trajectory $T$. The algorithm must maintain a tracking sequence $S$ that covers the portion of $T$ that is known so far, and after each update it must extend $S$ by adding additional pairs to it, without changing the pairs that have already been included. As has become standard for situations such as this one in which an online algorithm must make decisions without knowledge of the future, we measure the quality of an algorithm by its \emph{competitive ratio}. Specifically, if a deterministic online algorithm $A$ produces tracking sequence $S_A(T)$ from trajectory $T$, and the optimal tracking sequence is $S^*(T)$, then the competitive ratio of $A$ (for a given fixed set $\script D$ of regions) is
$$\sup_T\frac{|S_A(T)|}{|S^*(T)|}.$$
In the case of a randomized online algorithm, we measure the competitive ratio similarly, using the expected cost of the tracking sequence it generates. In this case, the competitive ratio is
$$\sup_T\frac{E[|S_A(T)|]}{|S^*(T)|}.$$

\maarten {I changed this back to the old definition of having $c$ disjoint sequences, because it's not the same thing: we want to allow newer additions to expire before older additions.}

As a variation of this problem, stemming  from trilateration problems in cellular phone network and sensor network coverage, we also consider the problem of finding tracking sequences with \emph{coverage} $c$. In this setting, we need to report a set of $c$ tracking sequences $S_1, S_2, \ldots, S_c$ for $T$ that are \emph {mutually disjoint} at any point in time: if a region $D$ appears for a time interval $[\tau_i, \tau_{i+1}]$ in one sequence $S_k$ and a time interval $[\sigma_j, \sigma_{j+1}]$ in some other sequence $S_l$, we require that the intervals $[\tau_i, \tau_{i+1}]$ and $[\sigma_j, \sigma_{j+1}]$ are disjoint. We wish to minimize the total cost $\sum_{i=1}^c |S_i|$ of a set of tracking sequences with coverage $c$, and in both the offline and online versions of the problem.

\section {Offline Tracking} \label {sec:static}


Even though we focus on the case where the trajectories of the entities are not known in advance, we also study the offline tracking problem. We will use some of the observations and algorithms from the offline problem in our analysis of algorithms for the online problem.

    \tweeplaatjes[scale=0.75]{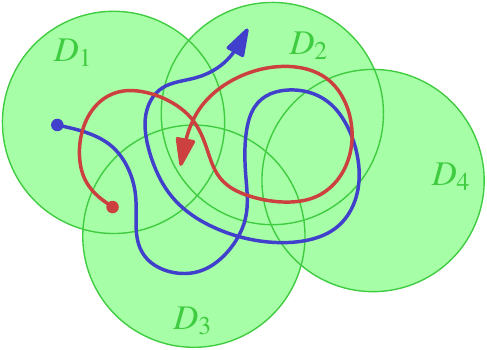} {ex-translation} {An example of of four sensors and two trajectories, in the original setting (a) and the corresponding interval representation (b).}
  
    As mentioned in Section~\ref {sec:dynamic}, we may view the input as a sequence of events that describe when an entity enters or leaves the region belonging to a sensor.
    In other words, we can translate the offline tracking problem into a problem with one continuous dimension (time), where for each sensor we represent the set of times during which the entity is inside or outside that sensor as a set of intervals in this time dimension.
    Figure~\ref {fig:ex-original+ex-translation} shows an example of two different trajectories travelling in the same set of regions $\script D$, and the corresponding interval representations of these trajectories.

  \subsection {Greedy Algorithm for Offline $c$-coverage Tracking}
    We describe a greedy algorithm for the offline problem where $T$ needs to be covered by $c$ disjoint sensors (trilateration) at any time. The original problem is a special case where $c=1$.
    
    In the greedy algorithm, We start with the $c$ longest available segments at the start.
    Now, whenever we reach the end of an interval at time $\tau$, we consider the set of available intervals (intervals that contain $\tau$, and are not currently in use by one of the other $c-1$ trackers), and always switch to the one that continues for the longest time into the future. 
    Figure~\ref {fig:greedy-1} shows a simple example for $c=1$.
  
    \begin {theorem} \label {thm:greedy}
    The greedy algorithm solves the offline tracking problem optimally, in polynomial time.
    \end {theorem}

    \begin {proof}
      Consider an arbitrary solution. First, we may assume whenever a path enters an interval, it stays there until the end of the interval. The only reason why it would leave is to make room for another path. But in that case, we could switch the roles of these two paths, leading to a better solution. Similarly, if a path reaches the end of an interval and does not jump to the longest available interval, this must be because another path uses that interval in the future. But then we can just select it anyway, and when the other path wants to select it, we send it instead to the place where our first path would have been at that time. 
      See Figure~\ref {fig:greedy-2}.
      This gives another solution that is at least as good.
      
      Assuming the input is given as a sequence of events, the greedy algorithm can be trivially implemented in time quadratic in the length of this sequence (this can likely be improved). Since we assume that the arrangement of the regions in $\script D$  has polynomial complexity and each piece of the trajectory has a constant number of intersections with the regions, the length of the event sequence is polynomial in $n$, the number of regions in $\script D$, and $m$, the number of pieces of $T$.
    \end {proof}

    \tweeplaatjes[scale=0.75]{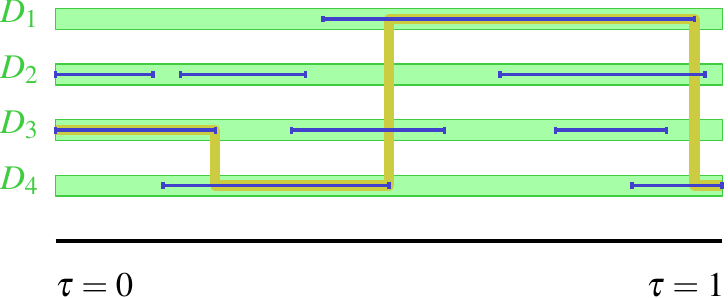} {greedy-2} {Greed is good!}

\section {Online Tracking} \label {sec:dynamic}

 We now move on to the dynamic setting. We assume that we are given the start locations of the trajectory, and receive a sequence of updates extending the trajectory. From these updates we can easily generate a sequence of \emph{events} caused when the trajectory crosses into or out of a region.
We will describe three algorithms for different settings, which are all based on the following observations.

  \maarten {This would be about the right moment to elaborate some more on competitive ratios.}

    \eenplaatje[scale=0.8]{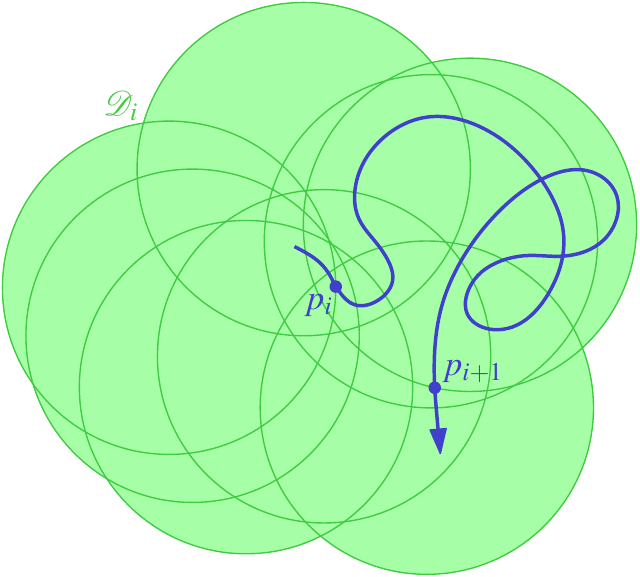} {The set $\script D_i$ of disks containing $p_i$, and the point $p_{i+1}$ where the trajectory leaves the last disk of $\script D_i$.}
    
    Let $T$ be the (unknown) trajectory of our moving entity, and recall that $T(\tau)$ denotes the point in space that the entity occupies at time $\tau$.
    Let $\tau_0$ be the starting time. We will define a sequence of times $\tau_i$ as follows. For any $i$, let $p_i = T(\tau_i)$ be the location of the entity at time $\tau_i$, and let $\script D_i \subset \script D$ be the set of regions that contain $p_i$. For each $D_{ij} \in \script D_i$, let $\tau'_{ij} $ be first the time after $\tau_i$ that the entity leaves $D_{ij}$.
    Now, let $\tau_{i+1} = \max_j \tau'_{ij}$ be the moment that the entity leaves the last of the regions in $\script D_i$ (note that it may have re-entered some of the regions).
    Figure~\ref {fig:disk-collection} shows an example.
    Let $\tau_k$ be the last assigned time (that is, the entity does not leave all disks $\script D_k$ before the the end of its trajectory).

    \begin {observation} \label {obs:k}
      Any tracking sequence $S$ for trajectory $T$ must have length at least $k$.
    \end {observation}
    
    \maarten {Not sure this needs a proof...}
    \begin {proof}
      For any $i$, a solution must have a region of $\script D_i$ at time $\tau_i$. However, since by construction there is no region that spans the entire time interval $[\tau_i, \tau_{i+1}+\eps]$ (for any $\eps > 0$), there must be at least one handover during this time, resulting in at least $k-1$ handovers, and at least $k$ regions.
    \end {proof}

    \subsection {Randomized Tracking with Logarithmic Competitive Ratio}
      With this terminology in place, we are now ready to describe our randomized algorithm.
      We begin by computing $\tau_0$, $p_0$ and $\script D_0$ at the start of $T$. 
      We will keep track of a set of candidate regions $\script C$, which we initialize to $\script C = \script D_0$, and select a random element from the candidate set as the first region to track the entity.      
      Whenever the trajectory leaves its currently assigned region, we compute the subset $\script C \subset \script D_i$ of all regions that contain the whole
      trajectory from $p_i$ to the event point, and if $\script C$ is not empty
      we select a new region randomly from $\script C$.
      When $\script C$ becomes empty, we have found the next point $p_{i+1}$, giving us a new nonempty candidate set $\script C$. Intuitively, for each point $p_i$, if the set of candidate regions containing $p_i$ is ordered by their exit times, the selected regions form a random increasing subsequence of this ordering, which has expected length $O(\log\ply)$, whereas the optimal algorithm incurs a cost of one for each point $p_i$.
      Refer to Algorithm~\ref {alg:random} for a more formal description of the algorithm.

      \begin {algorithm} [ht]
        \caption {Randomized online tracking algorithm.} \label {alg:random}
        We keep global variables $i$, $\script C$, and $S$.\\
        Initialization:
        \begin {compactenum}
          \item set $i = -1$
          \item call $\nextstep (\tau_0)$
          \item call $\picksensor (\tau_0)$
        \end {compactenum}
        Procedure $\nextstep (\tau)$:
        \begin {compactenum}
          \item increment $i$
          \item set $\tau_i = \tau$
          \item compute $p_i$ and $D_i$ (see Section~\ref {sec:runtime} for efficiency considerations)
          \item set $\script C = D_i$
        \end {compactenum}
        Procedure $\picksensor (\tau)$:
        \begin {compactenum}
          \item take a random element $C \in \script C$
          \item append $(\tau, C)$ to $S$
        \end {compactenum}
        Handle event $(\tau, D)$:
        \begin {compactenum}
          \item if the event is a region-enter-event, ignore it
          \item if $D \in \script C$ then
          \begin {compactenum}
            \item set $\script C = \script C \setminus \{D\}$.
            \item if $\script C$ now is empty, then call $\nextstep (\tau)$
            \item if $D$ is equal to the last region in $S$, then call $\picksensor (\tau)$
          \end {compactenum}         
        \end {compactenum}
        When there are no more events, output $S$.
      \end{algorithm}

      \begin {lemma} \label {lem:randomlog}
        Algorithm~\ref {alg:random} produces a valid solution of expected length $O (k \log \ply)$.
      \end {lemma}
      
      \begin {proof}
        By construction, we produce a new time stamp $\tau_i$ as soon as the entity leaves all available sensors that contain $p_{i-1}$. This corresponds exactly to the times we switch sensors in the greedy algorithm of Section~\ref {sec:static}, which by Theorem~\ref {thm:greedy} yields an optimal solution.
        
        Next, we prove that between $\tau_i$ and $\tau_{i+1}$ the expected number of new regions is $O (\log \ply)$.
        Recall that, for each $D_{ij} \in \script D_i$, we defined $\tau'_{ij}$ to be the first time after $\tau_i$ that the entity leaves $D_{ij}$.
        These numbers $\tau'_{ij}$ form a set of at most $|\script D_i|$ numbers. At each call to $\picksensor$, we select a random number from this set that is larger than any number we chose before.
        The expected length of such a sequence is $\ln |\script D_i| +O(1)=O(\log \ply)$.
      \end {proof}

      Combining Observation~\ref {obs:k} and Lemma~\ref {lem:randomlog}, we see
      that Algorithm~\ref {alg:random} has a competitive ratio of $O (\log \ply)$.

    \subsection {Deterministic Tracking with Linear Competitive Ratio}
      We now describe a deterministic variant of Algorithm~\ref {alg:random}. The only thing we change is
 that, instead of selecting a random member of the set $\script C$ of candidate regions, we select an arbitrary element of this set.   Here we assume that $\script C$ is represented in some deterministic way that we make no further assumptions about. For example, if the elements in $\script D$ are unit disks we might store them as a sorted list by the $x$-coordinate of their center points. Algorithm~\ref {alg:determ}  shows the pseudocode for the changed procedure.
         
      \begin {algorithm} [ht]
        \caption {Deterministic online tracking algorithm.} \label {alg:determ}
        Procedure $\picksensor (\tau)$:
        \begin {compactenum}
          \item let $C$ be the first element in $\script C$
          \item append $(\tau, C)$ to $S$
        \end {compactenum}
      \end{algorithm}

      This strategy may seem rather na\"ive, and indeed produces a competitive ratio that is exponentially larger than that of the randomized strategy of the previous section. But we will see in Section~\ref {sec:lowerbounds} that this is unavoidable, even for the specific case of unit disks.
      
      \begin {lemma} \label {lem:determlin}
        Algorithm~\ref {alg:determ} produces a valid solution of length $O (k \ply)$.
      \end {lemma}
      
      \begin {proof}
        Since the change in the algorithm does not influence the validity of the solution, the correctness of the algorithm follows directly from Lemma~\ref {lem:randomlog}.
        The length of a solution is clearly no more than $k \ply$, since there are $k$ steps and at each step there are at most $\ply$ disks in $\script D_i$.
      \end {proof}

      As before, combining Observation~\ref {obs:k} and Lemma~\ref {lem:determlin}, we see that Algorithm~\ref {alg:determ} has a competitive ratio of $O (\ply)$.
       
 \subsection {Deterministic Tracking in One Dimension}
In the 1-dimensional case, a better deterministic algorithm is possible. In this case, the regions of $\script D$ can only be connected intervals, due to our assumptions that they are closed connected subsets of $\R$.
            
Now, when we want to pick a new sensor, we have to choose between $c=|\script C|$ intervals that all contain the current position of the entity. For each interval $C_i$, let $\ell_i$ be the number of intervals in $\script C\setminus\{C_i\}$ that contain the left endpoint of $C_i$, and let $r_i$ be the number of intervals in $\script C\setminus\{C_i\}$ that contain the left endpoint of $C_i$. We say that an interval $C_i$ is \emph{good} if $\max(\ell_i,r_i)\le c/2$. Our deterministic algorithm simply chooses a good sensor at each step. Figure~\ref {fig:interval-determ} illustrates this.
      
      \begin {algorithm} [ht]
        \caption {Deterministic online tracking algorithm for $d=1$.} \label {alg:determ1}
        Procedure $\picksensor (\tau)$:
        \begin {compactenum}
          \item let $L$ be the sequence of left end points of the intervals in $\script C$, sorted from left to right
          \item let $R$ be the sequence of right end points of the intervals in $\script C$, sorted from right to left
          \item for each $C \in \script C$, let $i_C$ be the highest index of $C$ in either $L$ or $R$
          \item let $C^*$ be the sensor that has the lowest $i_C$
          \item append $(\tau, C^*)$ to $S$
        \end {compactenum}
      \end{algorithm}      
      
      The new algorithm is described in Algorithm~\ref {alg:determ1}.
As with our deterministic algorithm in higher dimensions, the only change from  Algorithm~\ref {alg:random} is the implementation of the $\picksensor$ procedure.
      \eenplaatje[scale=0.8]{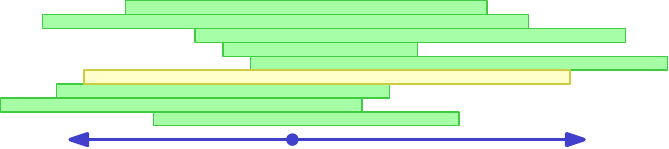} {A set of $8$ intervals covering the current location of the entity (blue dot). A good interval is highlighted; this interval has $\ell_i=3\le 8/2$ and $r_i=2\le 8/2$.}

      \begin {lemma} \label {lem:determlog}
        Algorithm~\ref {alg:determ1} produces a valid solution of length $O (k \log \ply)$.
      \end {lemma}
      
      \begin {proof}
       There always exists a good interval, by the pigeonhole principle, because there are at most $(c-1)/2$ intervals that are not good due to $\ell_i$ being too high and at most $(c-1)/2$ intervals that are not good due to $r_i$ being too high. Therefore the algorithm always succeeds in finding a good interval to choose.
        As in the previous section, the change in the algorithm does not influence the validity of the solution, so the correctness of the algorithm follows directly from Lemma~\ref {lem:randomlog}.

        Each time Algorithm~\ref{alg:determ1} performs a handover, it must be the case the trajectory has just crossed either the left endpoint or the right endpoint of the interval it most recently selected. Therefore, within the time interval from $\tau_i$ to $\tau_{i+1}$, the number of intervals in $\script C$ goes down by at least a factor of two at each handover, and it begins as at most $\ply$. Therefore, the number of handovers within this interval is at most $\log_2\ply$ and the total cost of the solution is at most $k\log_2\ply$.
      \end {proof}            
      
      Combining Observation~\ref {obs:k} and Lemma~\ref {lem:determlog}, we conclude that Algorithm~\ref {alg:determ1} also has a competitive ratio of $O (\log \ply)$.

    \subsection {Summary of Algorithms} \label {sec:runtime}
Our input assumptions ensure that any trajectory can be transformed in polynomial time into a sequence of events: trivially, for each piece in the piecewise description of the trajectory, we can determine the events involving that piece in time $O(n)$ (where $n=|\script D|$) and sort them in time $O(n\log n)$.

Once this sequence is known, it is straightforward to maintain both the set of regions containing the current endpoint of the trajectory, and the set $\script C$ of candidate regions, in constant time per event. Additionally, each event may cause our algorithms to select a new region, which may in each case be performed given the set $\script C$ in time $O(|\script C|)=O(\ply)$. Therefore, if there are $m$ events in the sequence, the running time of our algorithms (once the event sequence is known) is at most $O(m\ply)$.

Additionally, geometric data structures (such as those for point location among fat objects~\cite {os-rsplfo-96}) may be of use in more quickly  finding the sequence of events, or for more quickly selecting a region from $\script C$; we have not carefully analyzed these possibilities, as our focus is primarily on the competitive ratio of our algorithms rather than on their running times.

We summarize these results in the following theorem:

      \begin {theorem}
        Given a set $\script D$ of $n$ connected regions in $\R^d$, and a trajectory $T$,
        \begin {itemize}
          \item there is a randomized strategy for the online tracking problem that achieves a competitive ratio of $O (\log \ply)$; and
          \item there are deterministic strategies for  the online tracking problem that achieve a competitive ratio of $O (\log \ply)$ when $d = 1$ or $O (\ply)$ when $d > 1$.
        \end {itemize}
        Each of these strategies may be implemented in polynomial time.
      \end {theorem}

\section {Lower Bounds} \label {sec:lowerbounds}
  
We now provide several lower bounds on the best competitive ratio that any deterministic or randomized algorithm can hope to achieve. Our lower bounds use only very simple regions in $\script D$: similar rhombi, in one case, unit disks in $\R^d$ in a second case, and unit intervals in $\R$ in the third case. These bounds show that our algorithms are optimal, even with strong additional assumptions about the shapes of the regions.

\subsection{Lower Bounds on Stateless Algorithms}
\eenplaatje[scale=0.75]{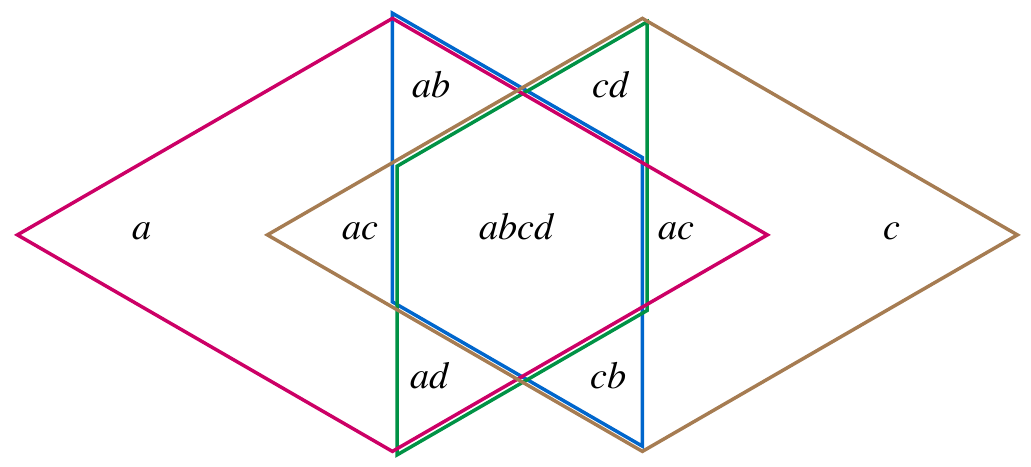} {Four similar rhombi form a set of regions for which no stateless algorithm can be competitive.}

An algorithm is \emph{stateless} if the next sensor that covers the moving point, when it moves out of range of its current sensor, is a function only of its location and not of its previous state or its history of motion. Because they do not need to store and retrieve as much information, stateless algorithms provide a very enticing possibility for the solution of the online tracking problem, but as we show in this section, they cannot provide a competitive solution.

\begin{theorem}  \label {thm:lowerbound_stateless}
There exists a set $\script D$ of four similar rhombi in $\R^2$, such that any stateless algorithm for the online tracking problem has unbounded competitive ratio.
\end{theorem}

\begin{proof}
The set $\script D$ is shown in Figure~\ref{fig:rhombi}. It consists of four rhombi $a$, $b$, $c$, and $d$; these rhombi partition the plane into regions (labeled in the figure by the rhombi containing them) such that the common intersection $abcd$ of the rhombi is directly adjacent to regions labeled $ab$, $ac$, $ad$, $bc$, and $cd$.

Let $G$ be a graph that has the four rhombi as its vertices, and the five pairs $ab$, $ac$, $ad$, $bc$, and $cd$ as its edges. Let $A$ be a stateless algorithm for $\script D$, and orient the edge $xy$ of $G$ from $x$ to $y$ if it is possible for algorithm $A$ to choose region $y$ when it performs a handover for a trajectory that moves from region $abcd$ to region $xy$. If different trajectories would cause $A$ to choose either $x$ or $y$, orient edge $xy$ arbitrarily.

Because $G$ has four vertices and five edges, by the pigeonhole principle there must be some vertex $x$ with two outward-oriented edges $xy$ and $xz$. There exists a trajectory $T$ that repeatedly passes from region $abcd$ to $xy$, back to $abcd$, to $xz$, and back to $abcd$, such that on each repetition algorithm $A$ performs two handovers, from $z$ to $y$ and back to $z$. However, the optimal strategy for trajectory $T$ is to cover the entire trajectory with region $x$, performing no handovers. Therefore, algorithm $A$ has unbounded competitive ratio.
\end{proof}

\subsection{Lower Bounds on Deterministic Algorithms}
\tweeplaatjes[scale=0.75]{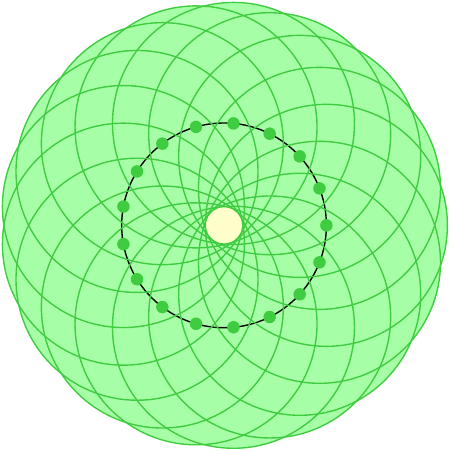} {flower-heart} {(a) A set of $\ply$ disks whose centers are equally spaced on a circle. (b) The heart of the construction, zoomed in. The yellow cell is inside all disks; the red cells are inside all but one disk.}

Next, we show that any deterministic algorithm in two or more dimensions must have a competitive ratio of $\ply$ or larger, matching our deterministic upper bound and exponentially worse than our randomized upper bound.  The lower bound construction consists of a set of $\ply$ unit disks with their centers on a circle, all containing a common point (Figure~\ref {fig:flower+flower-heart}). The idea is that if the trajectory starts at this common point, it can exit from any single disk, in particular, the one that a deterministic algorithm previously chose.

\begin {theorem} \label {thm:lowerbound_determ}
There exists a set $\script D$ of unit disks in $\R^2$, such that any deterministic algorithm for the online tracking problem has competitive ratio at least $\ply-1$.
\end {theorem}

\begin {proof}
Let $\script D$ be a set of $\ply$ unit disks whose centers are equally spaced on a given circle $C$ of slightly less than unit radius, as in Figure~\ref {fig:flower}. Let the moving point to be tracked start at the point $p_0$ at the center of $C$, in the common interior of all disks. For each disk $D_i \in \script D$, there exists a cell $X_i$ in the arrangement that is interior to all disks in $\script D \setminus\{D_i\}$, but outside $D_i$ itself. Furthermore, this cell is directly adjacent to the center cell. See Figure~\ref {fig:flower-heart} for an illustration.

Now, let $A$ be any deterministic algorithm for the online tracking problem, and construct a sequence of updates to trajectory $T$ as follows. Initially, $T$ consists only of the single point $p_0$. At each step, let algorithm $A$ update its tracking sequence to cover the current trajectory, let $D_i$ be the final region in the tracking sequence constructed by algorithm $A$, and then update the trajectory to include a path to $X_i$ and back to the center.

Since $X_i$ is not covered by $D_i$, algorithm $A$ must increase the cost of its tracking sequence by at least one after every update. That is, $|S_A(T)|\ge|T|$. However, in the optimal tracking sequence, every $\ply-1$ consecutive updates can be covered by a single region $D_i$, so $S^*(T)\le|T|/(\ply-1)$. Therefore, the competitive ratio of $A$ is at least $\ply-1$.
\end {proof}
      
This construction generalizes to any $d > 2$.

\subsection{Lower Bounds on Randomized Algorithms} \label {sec:lowerbounds_random}
The above lower bound construction uses the fact that the algorithm to solve the problem is deterministic: an adversary constructs a tracking sequence by reacting to each decision made by the algorithm. For a randomized algorithm, this is not allowed. Instead, the adversary must
select an entire input sequence, knowing the algorithm but not knowing the random choices to be made by the algorithm. Once this selection is made, we compare the quality of the solution produced by the randomized algorithm to the optimal solution.
By Yao's principle~\cite{cllr-bbcr-97,Yao-FOCS-77}, 
finding a randomized lower bound in this model is equivalent to finding a random distribution $R$ on the set of possible update sequences such that, for every possible deterministic algorithm $A$, the expected value of the competitive ratio of $A$ on a sequence from $R$ is high.

\eenplaatje[scale=0.8]{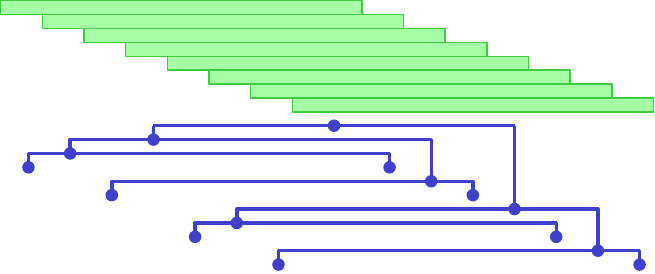} {A set of $\ply = 8$ intervals, and a tree of $8$ different trajectories in $\R^1$ (horizontal dimension).}

Our lower bound construction consists of $\ply$ unit intervals that contain a common point, and a tree of $\ply$ different possible paths for the moving object to take, each of which leaves the intervals in a different ordering, in a binary tree-like fashion. Half of the trajectories start by going to the left until they are outside the right half of the intervals, the others start towards the right until they are outside the left half of the intervals, and this recurses, as shown in Figure~\ref {fig:interval-tree}.

More formally, let us assume for simplicity that $\ply$ is a power of $2$.
Let $\script D$ be a set of $\ply$ distinct unit intervals in $\R$, containing a common point $p_0$.
For any $k \in [1, \ply]$ we define point $p_k$ to be a point outside the leftmost $k$ intervals but in the interior of the rest, and $p_{-k}$ to be a point outside the rightmost $k$ intervals but in the interior of the rest.

Now, for each $j \in [1, \ply]$, we construct a trajectory $T_j$ with $h = \log \ply$ steps, as follows. We define an index $\xi(j,i)$ for all $j \in [1, \ply]$ and all $i \in [1, h]$ such that trajectory $T_j$ is at point $p_{\xi(j,i)}$ at step $i$. At step $0$, all trajectories start at $\xi(j,0) = 0$.  Then, at step $i$:
      \begin {itemize}
        \item all $T_j$ with $j \mod 2^{h-i} \leq 2^{h-i-1}$ move to the left to $\xi(j,i) = \min_{l < i} \xi(j,l) - 2^{h-i}$,
        \item all $T_j$ with $j \mod 2^{h-i} > 2^{h-i-1}$ move to the right to $\xi(j,i) = \max_{l < i} \xi(j,l) + 2^{h-i}$.
      \end {itemize}
Figure~\ref {fig:interval-tree} shows $\script T$ be the resulting set of these $\ply$ trajectories in a tree representation.

\begin {theorem} \label {thm:lowerbound_random}
There exists a set $\script D$ of unit intervals in $\R$, for which any randomized algorithm to solve the online tracking problem has competitive ratio $\Omega (\log \ply)$.
\end {theorem}
      
\begin {proof}
Let $\script D$ and the set of trajectories $\script T$ be as described above. Let $R$ be a probability distribution over the set of all possible trajectories that has a probability of $1 / \ply$ to be any element of $\script T$, and a probability of $0$ elsewhere.
        
Now, let $A$ be any deterministic algorithm for the online tracking problem.  At each level of the tree, each region $D_i$ that algorithm $A$ might have selected as the final region in its tracking sequence fails to cover one of the two points that the moving point could move to next, and each of these points is selected with probability $1/2$, so algorithm $A$ must extend its tracking sequence with probability $1/2$, and its expected cost on that level is $1/2$. The number of levels is $\log_2\ply$, so the total expected cost of algorithm $A$ is $1+\frac12\log_2\ply$, whereas the optimal cost on the same trajectory is $1$. Therefore the competitive ratio of algorithm $A$ on a random trajectory with distribution $R$ is at least $1+\frac12\log_2\ply$.

It follows by Yao's principle that the same value $1+\frac12\log_2\ply$ is also a lower bound on the competitive ratio of any randomized online tracking algorithm.
\end {proof}

Although the trajectories formed in this proof are short relative to the size of $\script D$, this is not an essential feature of the proof: by concatenating multiple trajectories drawn from the same distribution, we can find a random distribution on arbitrarily long trajectories leading to the same $1+\frac12\log_2\ply$ lower bound.
This construction generalizes to unit balls in any dimension $d > 1$ as well.

\section{Trilateration}

  We now extend our online algorithms to the scenario where the moving entity needs to be covered by $c$ different sensors at all times. Obviously, this means we need to strengthen our assumptions on $\script D$ and $T$ slightly: every point of $T[0,\infty)$ should be inside at least $c$ regions of $\script D$ for this to be possible. 
  
  We analyze the online version of the problem, and provide the competitive ratio as a function of $\ply$ and $c$. As $c$ tends to $\ply$, we may expect the competitive ratio to improve, since in the extreme case we simply need to use all available sensors and have no choice in the matter. Indeed, we show that a randomized algorithm exists that has competitive ratio $O (\log (\ply - c))$, and that again no better ratio is possible in this case.
  
  \subsection {Randomized Algorithm}
    \maarten {Ok, after thinking some more about it, it seems the algorithm extension is quite straightforward after all, it's just the analysis that gets a bit more complicated.}
    The randomized algorithm with expected competitive ratio of $\log \ply$ for the simple version of the online tracking problem can be extended to the case in which we want to track the entity with $c$ sensors.
    The reason is that the greedy algorithm still works for the offline problem, as we proved in Theorem~\ref {thm:greedy}.
    Interestingly, the competitive ratio gets better as $c$ increases:
    the algorithm can be adapted to achieve a competitive ratio of $O (\log (\ply-c))$, so when the required coverage is close to the ply this gives a constant competitive ratio. The description gets slightly more complicated. We will now prove this theorem.

    \eenplaatje[width=4.5in]{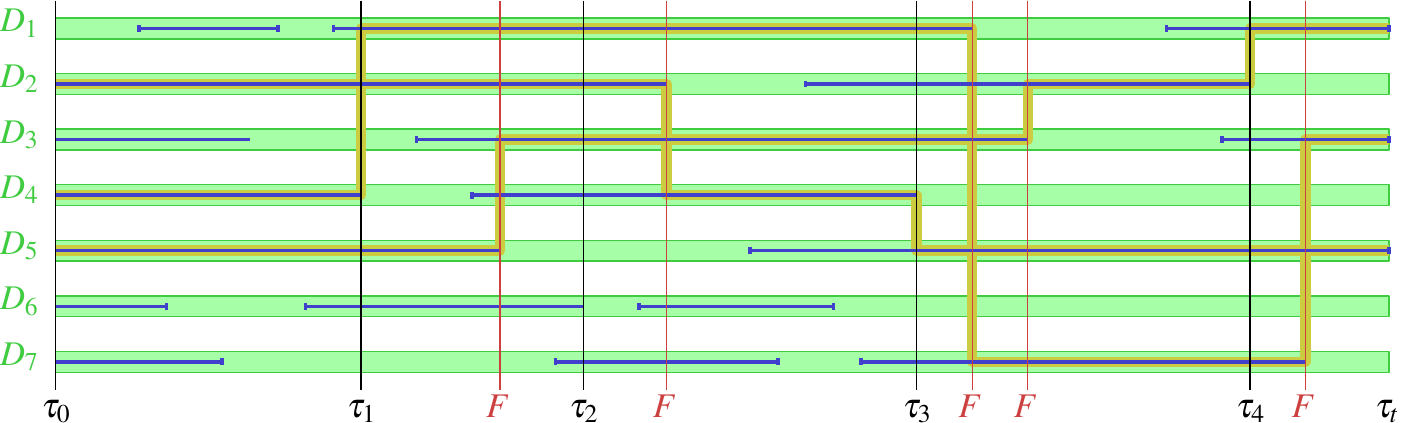} {Illustration of the two type of event times when covering the entity with $c=3$ sensors at the same time. The yellow paths indicate the optimal greedy solution. The times marked with $\tau_i$ are the times at which we start a new step in the algorithm. The times marked with $F$ are times when a fixed interval ends. Note that at time $\tau_2$ none of the sensors in the greedy solution changes its interval.}

\begin {theorem} \label {thm:trilateration-algorithm}
  There exists a randomized algorithm that solves the trilateration problem with a competitive ratio of $O (\log (\ply - c))$.
\end {theorem}

    First, we describe the altered algorithm.
    Again, we construct a sequence of time stamps $\tau_i$ with corresponding location of the entity $p_i$ and set of disks $\script D_i$ that contain $p_i$. Obviously, we need to assume $|\script D_i| \geq c$. Now, as before, we maintain a set of candidates $\script C$ that is initialized to $\script D_0$, and whenever the trajectory exits a sensor region, we remove that region from $\script C$ and replace it with a new random element from $\script C$ that is not currently used by another sensor. 
    If there are not enough regions left in $\script C$ (that is, if $|\script C| = c - 1$ after removing the current sensor), we mark the current time as $\tau_{i+1}$, and compute a new set of candidates (note that the $c-1$ remaining old candidates will also be in the new set of candidates). This proceeds until we mark $\tau_k$ as the end of time.
    \maarten {need to add more pseudocode?}

    Now, to analyze the performance of this algorithm, we will define an additional class of events. Just before we reach a new time $\tau_i$, we mark all $c$ sensors currently in use (which must correspond to the current candidate set $\script C$) as \emph {fixed}. Whenever a fixed sensor reaches the end of its interval, we say we have an $F$-event, and after the $F$-event we stop considering this sensor as being fixed. In particular, this means that each new $\tau_i$ corresponds to an $F$-event, but there could be more $F$-events. Figure~\ref {fig:greedy-events} shows the time stamps and $F$-events for a small example. If we define $F_i$ to be the number of $F$-events between $\tau_i$ and $\tau_{i+1}$, including the former and excluding the latter, then $1 \leq F_i \leq c$ for all $i$. We also write $m = \sum_{i=0}^k F_i$

    Now we will show that any solution needs $\Omega(m)$ handovers, and that our algorithm produces a solution with  an expected number of $O (m \log (\ply - c))$ handovers.

    \begin {lemma}
      The optimal solution to the tracking problem uses $\Theta(m)$ handovers.
    \end {lemma}

    \begin {proof}
      By Theorem~\ref {thm:greedy} the greedy approach yields an optimal solution, so we only need to bound the number of handovers in that solution. 

      First, we argue that each handover in the greedy solution coincides with an $F$-event. At the start, the greedy algorithm assigns the $c$ longest available itervals, which all become fixed before they run out.
      By construction, for any $i$ there must be a set of exactly $c$ sensors that cover the whole interval $[\tau_i, \tau_{i+1}]$. Whenever an $F$-event occurs in this interval and a new sensor has to be assigned, the greedy approach will pick the longest available one, and there must be at least one available interval that extends beyond $\tau_{i+1}$, which means it will also become fixed before it runs out.

      Not all $F$-events need to be used by the greedy solution however (as can be seen in Figure~\ref {fig:greedy-events}), so we will now prove that at least half of them are, by charging the unused events to used ones. Suppose an $F$-event is not used, and assume it occurs in the interval $[\tau_i, \tau_{i+1})$. The fact that the sensor became fixed means its interval started already at or before time $\tau_{i-1}$. Suppose there are $f$ unused $F$-events in the time interval $[\tau_i, \tau_{i+1})$. Since there are only $c$ sensors that completely cover the time interval $[\tau_{i-1}, \tau_{i}]$, and $f$ of them are not used, there must have been $f$ handovers in the time interval $[\tau_{i-1}, \tau_{i}]$, which by the previous argument occured during $F$-events.        
    \end {proof}

    \begin {lemma}
      Our randomized solution to the online tracking problem produces a valid solution of expected length $O(m \log (\ply - c))$.
    \end {lemma}

    \begin {proof}
      There are $F_i$ $F$-events during the time interval $[\tau_i, \tau_{i+1})$. When each of these events occurs, we pick a random element from the set of candidates $\script D_i$, that is larger than any other random element we chose so far. We know that $c \leq |\script D_i| \leq \ply$. We also know that $c - F_i$ sensors stay on their intervals during this time, so the set of free candidates is in fact at most $\ply - c + F_i$.
      This results in an expected number of $F_i + \log ((\ply - c + F_i)/(F_i))$ changes (because we first take $F_i$ random elements, and then a random increasing sequence that starts larger than any of these numbers). 

      The expected total number of handovers is $\sum_{i=0}^k (F_i + \log ((\ply - c + F_i)/(F_i)))$, which in the worst case (occuring when $F_i = 1$ for all $i$) comes down to $m \log (\ply - c)$.
    \end {proof}      

    These two lemmas together imply the competitive ratio we claimed.

    \maarten {The other algorithms are easier anyway, so they probably also generalize.}

  \subsection {Lower Bound for the Randomized Case}
    We can extend the construction of Section~\ref {sec:lowerbounds_random} to the current situation, although we now require the regions of $\script D$ to be intervals of two different lengths. Note that when $d > 1$, this is no real restriction since we can always position a set of unit disks in such a way that their intersections with a given line form intervals of two different lengths.
    
    We define a set of intervals $\script D$, consisting of $\ply - c + 1$ unit intervals that form the construction depicted in Figure~\ref {fig:interval-tree}, and additionally $c - 1$ intervals of length $2$ that cover all $\ply - c + 1$ unit intervals.
    
\begin {theorem} \label {thm:trilateration-lowerbound}
  There exists a set $\script D$ of intervals in $\R$ of two different lengths, for which any randomized algorithm to solve the online tracking problem has competitive ratio $\Omega (\log (\ply - c))$.
\end {theorem}

    \begin {proof}
      A solution of $c$ disjoint tracking sequences must at any time use at least one of the $\ply - c + 1$ unit intervals, since there are only $c-1$ other intervals.
      We may assume that there is one single sequence that only uses unit intervals, since otherwise we could swap pieces of the sequences. But by Theorem~\ref {thm:lowerbound_random}, no algorithm can produce such a sequence with less than $\Omega (\log h)$ expected handovers, if there are $h$ unit intervals in the construction. Since in our case $h = \ply - c + 1$, the result follows.
    \end {proof}

Our other lower bounds can also be extended just as easily.

\section {Conclusions}

  We studied the online problem of tracking a moving entity among sensors with a minimal number of handovers, combining the kinetic data and online algorithms paradigms.
  We provided several algorithms with optimal competitive ratios. Interestingly, randomized strategies are able to provably perform significantly better than deterministic strategies, and arbitrarily better than stateless strategies (which form a very natural and attractive class of algortihms in our application).
  
  We are able to track multiple entities using the same algorithms, by simply treating them independently. As a future direction of research, it would be interesting to study the situation where each sensor has a maximum capacity $C$, and cannot track more than $C$ different entities at the same time.
  Another possible direction of research is to analyze and optimize the running times of our strategies for particular classes of region shapes or trajectories, something we have made no attempt at.


{\small\raggedright
\bibliographystyle {abbrv}
\bibliography {refs,../../../bibliographies/geom}
}

\end{document}